\documentclass[11pt]{article}

\usepackage{a4}
\usepackage{amsmath}
\usepackage{amsfonts}
\usepackage{amssymb}
\usepackage{euscript}

\newtheorem{theorem}{Theorem}

\newtheorem{lemma}[theorem]{Lemma}

\newtheorem{proposition}[theorem]{Proposition}

\newenvironment{proof}{\noindent {\it Proof:~}\ }{~\rule{1mm}{2mm}\medskip}
\newenvironment{proofof}[2]{\noindent {\it Proof of #1}~#2: \ }{~\rule{1mm}{2mm}\medskip}

\def\transp#1{\,{^t\hspace{-1pt}#1}}
\def\esp#1{{\rm E}\left[#1\right]}
\def\prob#1{{\rm P}\!\left(#1\right)}
\def\wt#1{{\rm wt}\left(#1\right)}

\newcommand{\Z}{\mathbb Z}
\newcommand{\F}{\mathbb F}
\renewcommand{\H}{\mathbf H}
\newcommand{\A}{\mathbf A}
\newcommand{\B}{\mathbf B}
\newcommand{\I}{\mathbf I}
\newcommand{\X}{\mathbf X}
\newcommand{\rand}{C_{\rm rand}}
\newcommand{\C}{\EuScript C}

\newcommand{\x}{\mathbf x}
\newcommand{\y}{\mathbf y}
\newcommand{\s}{\mathbf s}
\renewcommand{\a}{\mathbf a}
\renewcommand{\u}{\mathbf u}
\renewcommand{\epsilon}{\varepsilon}

\begin{document}

\title{ Asymptotic improvement of the Gilbert-Varshamov bound for linear codes} 
\author{Philippe Gaborit\thanks{XLIM, Universit\'e de Limoges,
123, Av. Albert Thomas,
 87000 Limoges, France. {\tt gaborit@unilim.fr}}
\and
Gilles Z\'emor\thanks{Universit\'e de Bordeaux 1, Institut de Math\'ematiques
de Bordeaux, 351 cours de la Lib\'eration, 33405 Talence. 
{\tt
zemor@math.u-bordeaux1.fr} }
}

\date{August 29, 2007}
\maketitle
\begin{abstract} The Gilbert-Varshamov bound states that 
the maximum size $A_2(n,d)$  of a binary code of length $n$
and minimum distance $d$ satisfies
$A_2(n,d) \ge 2^n/V(n,d-1)$ where $V(n,d)=\sum_{i=0}^d \binom{n}{i}$
stands for the volume of a Hamming ball of radius $d$.
Recently Jiang and Vardy showed that for binary non-linear codes this bound
can be improved to $$A_2(n,d) \ge cn\frac{2^n}{V(n,d-1)}$$ for $c$
a constant and $d/n \le 0.499$. In this paper we show that  
certain asymptotic families of linear binary $[n,n/2]$ random double circulant
codes satisfy the same improved Gilbert-Varshamov bound. These 
results were partially presented at ISIT 2006 \cite{gz06}. 
\end{abstract}


\section{Introduction}
The Gilbert-Varshamov bound asserts that the maximum size $A_q(n,d)$
of a $q$-ary code of length $n$ and minimum Hamming distance $d$
satisfies
\begin{equation}
  \label{eq:VG}
  A_q(n,d) \geq \frac{q^n}{\sum_{i=0}^{d-1}\binom{n}{i}(q-1)^i}.
\end{equation}
This result
is certainly one of the most well-known in coding theory,
it was originally stated in 1952 by Gilbert \cite{g52}
and improved by Varshamov in \cite{va57}. In 1982 Tsfasman,
Vladuts and Zink \cite{tvz82} improved
the GV bound on the number of codewords 
by an exponential factor in the block length, but this spectacular
result only holds for some classes of non-binary codes.
Recently Jiang and Vardy \cite{jv04} improved the GV bound for 
non-linear binary codes by a linear factor in the block length $n$ to
\begin{equation}
\label{eq:GV+}
A_2(n,d) \ge cn\frac{2^n}{V(n,d-1)},
 \end{equation}
 for $d/n \le 0.499$, for a constant $c$ that depends only on
the ratio $d/n$ and where $V(n,d)=\sum_{i=0}^d \binom{n}{i}$
stands for the volume of a Hamming ball of radius $d$.
This new bound asymptotically surpasses previous improvements of
the binary Gilbert-Varshamov bound 
which  only managed to  multiply the right hand side in \eqref{eq:VG}
by a constant 
(see \cite{jv04} for references).
The method used by Jiang and Vardy relies on a graph-theoretic framework
and more specifically on locally sparse graphs which are used to yield
families of non-linear codes (their result was later slighlty improved
in \cite{vw05}).
In this paper we also improve on the the Gilbert-Varshamov bound
by a linear factor in the block length but for {\bf linear} codes,
thereby solving one of the open problems of \cite{jv04}. 
The method we use
is not related to graph theory and relies on double circulant random
codes.

Double circulant codes are $[2n,n]$ codes which are stable
under the action of permutations composed of two
circular permutations of order $n$ acting 
simultaneously on two differents halves of the coordinate set.
These codes can also be seen as quasi-cyclic codes, a natural generalization
of cyclic codes \cite{MW}.
Their study started in 1969 in \cite{kar69} and since they gave some
very good codes it was natural to wonder whether they
could be made to satisfy the Gilbert-Varshamov bound. 
A first step in that direction was made by Chen, Peterson and Weldon 
in~\cite{cpw69} who prove that when $2$ is a primitive
root of the ring $\Z/p\Z$ for $p$ a prime, double circulant 
$[2p,p]$ random codes satisfy the Gilbert-Varshamov bound;
unfortunately it is still unknown (this is Artin's celebrated conjecture, 1927)
whether an infinity of such $p$ exists.
Later Kasami \cite{kas74}, building on this idea, extended 
the result of~\cite{cpw69} to the case of powers of such $p$, 
and obtained a bound which is worse than the Gilbert-Varshamov bound
by an exponential factor in the block length (though a very small one).
 Later Kasami's work
was generalized to other cases in \cite{kab77,kro77,ks95},
and, in particular in \cite{che92}, bounds were proven for certain
classes of quasi-cyclic codes that are worse than the
Gilbert-Varshamov bound only by a subexponential factor in the block length.
In this paper, building anew on Kasami's idea we prove,
by using a probabilistic approach, that randomly chosen double circulant 
codes not only satisfy the
Gilbert-Varshamov bound with high probability,
but also the same linear improvement as that of
Jiang and Vardy (\ref{eq:GV+}). 

\medskip

The paper is organized as follows: in Section \ref{sec:overview}
we cover the main ideas involved.
We start by recalling the probabilistic method for deriving lower
bounds on the minimum distance of linear codes (section~\ref{sec:GV}),
then we introduce double circulant codes in section~\ref{sec:double}
and derive \eqref{eq:orbits2} an upperbound on the probability that
a random double circulant code contains a non-zero vector of weight not more
than a given $w$.
In section~\ref{sec:behaviour} we study the probability that a given
vector belongs to a randomly chosen double circulant code. Finally
in section~\ref{sec:simple} we derive our improved lower bound on the
minimum distance in the simple case when the codelength is $2p$
and $2$ is a primitive root of $\Z/p\Z$: the result is given in
Theorem~\ref{th:simple}.

in Section~\ref{sec:infinite}, we develop our method in the 
more complicated case of blocklengths $2p^m$, $p$ a ``Kasami'' prime,
in order to obtain an infinite family of double circulant codes with
an improved minimum distance.
Section~\ref{sec:preview} starts by giving an informal sketch of the
content of section~\ref{sec:infinite}, which is intended to give some
guidance to the reader and discuss the technical issues involved.
Section~\ref{sec:reducing} shows how to derive our main result,
which is Theorem~\ref{th:main}, from a proposition on the
weight distribution of a certain class of cyclic codes.
Finally section~\ref{sec:proof} is devoted to a proof of this last
proposition.

Section \ref{sec:more} concludes by some comments and side results.

\section{Overview of the method, the simple cases}\label{sec:overview}
\setcounter{subsection}{-1}
\subsection{The Gilbert Varshamov bound for linear codes and its
  improvement}
\label{sec:GV}
To put the rest of the paper into perspective and introduce notation, 
let us recall how the probabilistic method
derives the Gilbert Varshamov bound for linear codes.
Rather than bounding the code size from below by a function of the
minimum distance, as in \eqref{eq:GV+}, we fix a lower bound on the
code rate and find a lower bound on the minimum distance. We limit
ourselves to the rate $1/2$ case because it will be our main object of
study.

Let $\rand$ be the random code of length $2n$ and dimension $k\geq n$
obtained by choosing randomly and uniformly a $n\times 2n$ parity-check
matrix in $\{0,1\}^{n\times 2n}$. The probability that a given nonzero
vector
$\x = (x_1\ldots x_{2n})$ is a codeword is clearly $1/2^n$.
Let $w$ be a positive number, not necessarily an integer.
We are interested in the random variable $X(w)$ equal to the number
of nonzero codewords of $\rand$ of weight not more than $w$. In other
words
\begin{equation}
  \label{eq:X(w)}
  X(w) = \sum_{\x\in B_{2n}(w)}X_\x
\end{equation}
where $B_{2n}(w)$ denotes the set of nonzero vectors $\x$ of 
$V_{2n}=\{0,1\}^{2n}$
of weight at most $w$, and $X_\x$ is the Bernoulli random variable
equal to $1$ if $\x\in \rand$ and equal to zero otherwise. Now
whenever we prove that the probability $\prob{X(w)>0}$ is less
than~$1$, we prove the existence of a $[2n,k,d]$ code with
$k\geq n$ and $d>w$. Since the variable $X(w)$ is integer valued we
have
\begin{eqnarray*}
  \prob{X(w)>0} \leq \esp{X(w)} & = & \sum_{\x\in
    B_{2n}(w)}\esp{X_\x} = |B_{2n}(w)|\prob{\x\in\rand}\\
   & = & |B_{2n}(w)|\frac{1}{2^n}.
\end{eqnarray*}
Hence,  for every positive integers $n$ and $w$ satisfying
  $|B_{2n}(w)| < 2^n$
  there exists a linear code of parameters $[2n,n,d>w]$. Reworded,
  we have the following lower bound on $d$, essentially equivalent
  to \eqref{eq:VG}.
\begin{theorem}[GV bound]
  For every positive integer $n$ there exists a linear code of parameters
  $[2n,n,d]$ satisfying
  $$|B_{2n}(d)| \geq 2^n.$$
\end{theorem}
In the present paper we shall prove~:
\begin{theorem}\label{th:GVplus}
  There exists a positive constant $b$ and an infinite sequence of
  integers $n$ and $[2n,n,d]$ linear codes satisfying
  $$|B_{2n}(d)| \geq bn2^n.$$
\end{theorem}
This result, equivalent to \eqref{eq:GV+} for rate $1/2$, will be
obtained by again choosing random matrices, but from a
restricted class, namely the set of parity-check matrices of
double circulant codes.

\subsection{Double circulant codes}
\label{sec:double}
A binary {\em double circulant code} is a $[2n,n]$ linear code $C$ with a 
parity-check matrix of the form $\H = [\I_n\, |\,\A]$ where $\I_n$ is the
$n\times n$ identity matrix and 
$$\A = 
\begin{bmatrix}
       a_0 & a_{n-1} & \dots & a_1\\
       a_1 & a_0 & \dots & a_{2}\\
        a_{2} & a_{1} & \dots & a_{3}\\
       \hdotsfor{4} \\
       a_{n-1} & a_{n-2} & \dots & a_0
\end{bmatrix}.$$
There is a natural action of the group $\Z/n\Z$ on the space
$V_{2n}=\{0,1\}^{2n}$ of vectors $\x=(x_1\ldots x_n,x_{n+1}\ldots x_{2n})$
namely,
\begin{eqnarray*}
  \Z/n\Z\times V_{2n} & \rightarrow & V_{2n}\\
           (j,\x)     & \mapsto     & j\cdot\x 
\end{eqnarray*}
where
$$1\cdot\x =
      (x_{n},x_1\ldots x_{n-1},x_{2n},x_{n+1},\ldots x_{2n-1})$$
and $j\cdot\x = (j-1)\cdot(1\cdot\x)$.
The double circulant code $C$ is clearly invariant under this group action.
Consider now $C$ to be the random code $\rand$ obtained by choosing the
vector $\a=(a_0\ldots a_{n-1})$ randomly and uniformly in $\{0,1\}^n$.
As before, we are interested in the random variable $X(w)$ 
defined by \eqref{eq:X(w)} and equal to the number
of nonzero codewords of $\rand$ of weight not more than $w$. 
We are interested in the maximum value of $w$ for which we can claim
that $\prob{X(w)>0}<1$, for this will prove the existence 
of codes of parameters
$[2n,n,d> w]$. The core remark is now that, if $\y=j\cdot\x$, then
  $$X_\y = X_\x$$
where $X_\x$ ($X_\x$) is the Bernoulli random variable
equal to $1$ if $\x\in \rand$ ($\y\in \rand$) and equal to zero otherwise.
Let now $B_{2n}'(w)$ be a set of representatives of the orbits of the
elements of $B_{2n}(w)$, i.e. for any $\x\in B_{2n}(w)$, 
$|\{j\cdot \x ,j\in\Z/n\Z\}\cap B_{2n}'(w)|=1$. We clearly have $X(w)>0$
if and only if $X'(w)>0$ where
  $$X'(w) = \sum_{\x\in B_{2n}'(w)}X_\x.$$
Denote by $\ell(\x)$ the length (size) of the orbit of $\x$,
i.e. $\ell(\x)=\#\{j\cdot \x ,j\in\Z/n\Z\}$. We have
\begin{equation}
  \label{eq:orbits}
  X'(w) = \sum_{\x\in B_{2n}(w)}\frac{X_\x}{\ell(\x)}
\end{equation}
By writing $\prob{X(w)>0}=\prob{X'(w)>0}\leq \esp{X'(w)}$, together
with \eqref{eq:orbits} we obtain
\begin{equation}
  \label{eq:orbits2}
  \prob{X(w)>0}\leq\sum_{d|n}\sum_{\substack{\wt{\x}\leq w\\ \ell(\x)=d}}
  \frac{\esp{X_\x}}{d}.
\end{equation}
Suppose in particular
that $n$ is a prime, in that case orbits are of size $1$ or $n$,
and if $w<n$ then clearly the orbit of $\x$ has size $n$ for any
$\x\in B_{2n}(w)$, so that \eqref{eq:orbits2} becomes
  $$\prob{X(w)>0} \leq  \esp{X(w)}/n.$$
If we can manage to prove that
\begin{equation}
  \label{eq:expectation}
  \esp{X(w)} \leq |B_{2n}(w)|\frac{c}{2^n}
\end{equation}
for constant $c$,
then we will have proved the existence of double circulant codes of parameters 
$[2n,n,d> w]$, for any $w$ such that $|B_{2n}(w)|< \frac 1c n2^n$.

\subsection{The behaviour of $\prob{\x\in\rand}$}
\label{sec:behaviour}
To prove equality \eqref{eq:expectation}
we need to study carefully the quantities
$\esp{X_\x}$, for $\x\in B_{2n}(w)$, since
  $$\esp{X(w)} = \sum_{\x\in B_{2n}(w)}\esp{X_\x}.$$
For $\x\in V_{2n}$, let us write $\x = (\x_L, \x_R)$ with
$\x_L,\x_R\in \{0,1\}^n$. Consider the syndrome function $\sigma$
\begin{eqnarray*}
 \sigma : V_{2n} & \rightarrow & V_n\\
          \x     & \mapsto     & \sigma(\x) =\x\transp{\H}
                                 =\sigma_L(\x)+\sigma_R(\x)
\end{eqnarray*}
where $\sigma_L(\x)=\x_L$ and $\sigma_R(\x)=\x_R\transp{\A}$.

For any binary vector of length $n$, $\u=(u_0,\ldots ,u_{n-1})$,
denote by $\u(Z)=u_0 + u_1Z+\cdots +u_{n-1}Z^{n-1}$ its polynomial 
representation in the ring $\F_2[Z]/(Z^n+1)$. For any $\u\in V_n$,
let $C(\u)$ denote the cyclic code of length $n$ generated by the
polynomial representation $\u(Z)$ of $\u$.
Since $\sigma_R(\x)$ has polynomial representation
equal to $\x_R(Z)\a(Z)$, we obtain easily
\begin{lemma}\label{lem:Cx}
The right syndrome  $\sigma_R(\x)$ of any given $\x\in V_{2n}$
is uniformly distributed in the cyclic code $C(\x_R)$. Therefore, the
probability $\prob{\x\in\rand}$ that $\x$ is a codeword of the random
code $\rand$ is
$$\begin{array}{ll}
\bullet\hspace{2mm} \prob{\x\in\rand} = 1/|C(\x_R)| & 
                    \text{if}\hspace{2mm} \x_L\in C(\x_R),\\
\bullet\hspace{2mm}  \prob{\x\in\rand} = 0       & 
                    \text{if}\hspace{2mm} \x_L\not\in C(\x_R).
\end{array}$$
\end{lemma}

\subsection{The case $n$ prime and $2$ primitive modulo $n$}
\label{sec:simple}
If $n$ is prime and $2$ is primitive modulo $n$ then, over $\F_2[Z]$, 
the factorization of $Z^n+1$ into irreducible polynomials is
  $$Z^n+1 = (1+Z)(1+Z+Z^2+\cdots +Z^{n-1})$$
and there is only one non-trivial cyclic code of length $n$, namely
the $[n,n-1,2]$ even-weight code. Therefore 
$\prob{X(w)>0}=\prob{X'(w)>0}\leq \esp{X'(w)}$ together with
\eqref{eq:orbits} and Lemma \ref{lem:Cx} give
\begin{eqnarray}
  \prob{X(w)>0} &\leq& \sum_{\substack{\wt{\x_L}+\wt{\x_R}\leq w\\
  \wt{\x_R}\;\text{odd}}} \frac{1}{n2^n} + 
  \sum_{\substack{\wt{\x_L}+\wt{\x_R}\leq w\\ \wt{\x_R}\;\text{even}\\
  \wt{\x_L}\;\text{even}}}\frac{1}{n2^{n-1}}\label{eq:easy}\\
  \prob{X(w)>0} &\leq& 2|B_{2n}(w)|\frac{1}{n2^{n}}.\nonumber
\end{eqnarray}
We therefore have the following result:
\begin{theorem}\label{th:simple}
 If $p$ is prime and $2$ is primitive modulo $p$, then there exist
 double circulant codes of parameters $[2p,p,d>w]$ for any positive number $w$
 such that $$2|B_{2p}(w)|<p2^p.$$
\end{theorem}

Unfortunately, it is not known (though it is conjectured)
whether there exists an infinite family
of primes $p$ for which $2$ is primitive modulo $p$. Therefore, to
obtain Theorem~\ref{th:GVplus} we will envisage cases when $n$ is
non-prime. This will involve two technical difficulties, namely
dealing with non-trivial divisors $d$ of $n$ in \eqref{eq:orbits2},
and non-trivial cyclic codes $C(\x_R)$ of length $n$ in Lemma~\ref{lem:Cx}.

\section{An infinite family of double circulant
  codes}\label{sec:infinite}
\subsection{Preview}
\label{sec:preview}
In this section we will study the behaviour of the minimum distance
of random double circulant codes for the infinite 
sequences of blocklengths $2n$ introduced by Kasami~: 
we will have $n=p^m$ for suitably chosen $p$.
We will first specialise inequality \eqref{eq:orbits2} to this case,
for which all the possible orbit sizes $\ell$ are powers of $p$,
$p^s$, $s\leq m$. Applying Lemma~\ref{lem:Cx} will lead us to an
upper bound \eqref{eq:triplesum}
on $\prob{X(w)>0}$ that involves the weight distributions
of the cyclic codes of length $n$. This upper bound can be essentially
thought of as the same as \eqref{eq:easy}, plus a number of parasite
terms involving all vectors $\x=(\x_L,\x_R)$ of $B_{2n}(w)$ for which
both $\x_L$ and $\x_R$ are codewords of some cyclic code of length $n$
that is neither the whole space $\{0,1\}^n$ nor the $[n,n-1,2]$
even-weight subcode. The problem at hand is to control the parasite
terms so that they do not pollute too much the main term i.e. the
right hand side of \eqref{eq:easy}. To do this, the crucial part will be to
bound from above with enough precision terms of the form
\begin{equation}
  \label{eq:distribution}
  \sum_{i+j\leq w}A_i(C)A_j(C)\frac{1}{|C|}
\end{equation}
where $C$ is a cyclic code of length $n$ and $A_i(C)$ is the number of
codewords of weight $i$.
In section~\ref{sec:reducing} we shall state such an upper bound,
namely Proposition~\ref{prop:K}, and show how it leads to the desired
result which will be embodied by Theorem~\ref{th:main}.

Section~\ref{sec:proof} will then be devoted to proving 
Proposition~\ref{prop:K}. It is not easy in general to estimate
the weight distribution of cyclic codes that don't have extra
properties, but it turns out that for these particular code lengths of
the form $n=p^m$, all cyclic codes $C$ have a special degenerate structure.
Either $C$ consists of a collection of vectors of the form
$(x,x,\ldots ,x)$ where $x$ is a subvector of length $n/p$ and is
repeated $p$ times, or $C$ is the dual of such a code. 
Section~\ref{sec:reducing} will have reduced the problem to the latter
class of cyclic codes only. Ideally, we would like to claim that
the cyclic codes $C$ have a binomial distribution of weights, i.e.
$A_i(C)\approx \frac{|C|}{2^n}\binom{n}{i}$, however this is not
true, the cyclic codes $C$ have many more low-weight codewords than would
be dictated by the binomial distribution. The problem of the
unbalanced couples $(i,j)$, ($i$ small and $j$ large or vice versa)
in the sum \eqref{eq:distribution} is therefore dealt with by the trivial
upper bound $A_i(C)\leq\binom{n}{i}$~: Lemma~\ref{lem:enumeration}
will show that these terms account for a sufficiently small fraction
of $|B_{2n}(w)|/2^n$. Lemma~\ref{lem:kappa} is the central result of
section~\ref{sec:proof} which gives a more refined upper bound on
$A_i(C)$ for $i$ well enough separated from $0$, i.e. $i\geq\kappa n$
for constant positive $\kappa$.
Fortunately, we do not need $A_i(C)$ to be too close to the binomial
distribution, and the cruder upper bound of
Lemma~\ref{lem:enumeration} will suffice to derive Proposition~\ref{prop:K}.

\subsection{Reducing the problem to the study of the weight
  distribution of certain cyclic codes}\label{sec:reducing}
Following Kasami \cite{kas74}, let us consider $n$ of the form $n=p^m$ where
$2$ is primitive modulo $p$ and $2^{p-1}\neq 1 \bmod p^2$.
It will be implicit that all the primes $p$ considered in the 
remainder of section~\ref{sec:infinite} will satisfy this property.
Let us also suppose $m\geq 2$, since the case $m=1$ is covered by
Theorem~\ref{th:simple}.

It is known \cite{kas74} that
the irreducible factors of $Z^n+1$ in $\F_2[Z]$ are $1+Z$ together with
all the polynomials of the form 
\begin{equation}
  \label{eq:irreducible}
  1+Q(Z)+Q(Z)^2+\cdots Q(Z)^{p-1}
\end{equation}
for $Q(Z)=Z,=Z^p,Z^{p^2},\ldots ,Z^{p^{m-1}}$.

Since $n$ is a prime power, 
\eqref{eq:orbits2} gets rewritten through Lemma \ref{lem:Cx} as:
\begin{equation}
  \label{eq:ps}
  \prob{X(w)>0}\leq\sum_{s=1}^m\sum_{\substack{\wt{\x}\leq w\\
      \ell(\x)=p^s\\ C(\x_L)\subset C(\x_R)}}
  \frac{1}{p^s|C(\x_R)|}
\end{equation}

Note that
$\x\in V_{2n}$ has orbit length $\ell(\x) < n$ if and only if both $\x_L$ and $\x_R$
are made up of $p$ successive identical subvectors of length $n/p$. Equivalently
$\x_L$ and $\x_R$ each belong to the cyclic code generated by the polynomial
\begin{equation}
  \label{eq:repetition}
  P_n(Z)=1 + Z^{n/p} + Z^{2n/p} + \cdots + Z^{(p-1)n/p}.
\end{equation}
Let $\C_n$ denote the set of those cyclic codes of length $n$ whose
generator polynomial is {\em not} a multiple of $P_n(Z)$. 
All the other cyclic codes of length $n$ are obtained
by duplicating $p$ times some cyclic code of length $n/p$.
Therefore, for $s=m$, the inner sum in \eqref{eq:ps} can be bounded
from above by:
\begin{equation}
  \label{eq:recursively}
  \sum_{C\in\C_n}\sum_{i+j\leq w}A_i(C)A_j(C)\frac{1}{n|C|}
\end{equation}
where $A_i(C)$ denotes the number of codewords of $C$ of weight $i$.
Applying \eqref{eq:recursively} recursively, we obtain from \eqref{eq:ps}
\begin{equation}
  \label{eq:triplesum}
  \prob{X(w)>0}\leq\sum_{s=0}^{m-1}
  \sum_{C\in\C_{n/p^s}}\sum_{i+j\leq w/p^s}A_i(C)A_j(C)\frac{1}{|C|n/p^s}.
\end{equation}

We now proceed to evaluate the righthandside of \eqref{eq:triplesum}. 
The most technical part of our proof of Theorem~\ref{th:GVplus} is
contained in the following Proposition.

\begin{proposition}\label{prop:K}
  There exist positive constants $q,K$, $c_1$ and $\gamma<1$ such that, 
  for any $n=p^m$ with $p\geq q$, we have $|B_{2n}(2Kn)|\leq 2^n$ and
  for any positive real number $w$,
  $K\leq w/2n\leq 1/4$, and for any cyclic code $C$ of $\C_n$, we have
  $$\sum_{i+j\leq w}A_i(C)A_j(C)\frac{1}{|C|}\leq 
  c_1\frac{|B_{2n}(w)|}{2^n}\gamma^{n-\dim C}.$$
  Suitable numerical values of the constants are $q=14^3$, $K=0.1$,
  $\gamma=1/2^{1/5}$, $c_1=2^{6/5}$.
\end{proposition}

Before proving Proposition \ref{prop:K}, let us derive the consequences
on the probability $\prob{X(w)>0}$. 
That will lead us to our main result, namely 
Theorem~\ref{th:main}, the consequence of which is Theorem~\ref{th:GVplus}.
We have:
\begin{lemma}\label{lem:c2}
  There exists a constant $c_2$ such that, for any $n=p^m$, $p>q$,
  and for any $K\leq w/2n\leq 1/4$,
  $$\sum_{C\in\C_n}\sum_{i+j\leq w}A_i(C)A_j(C)\frac{1}{|C|}\leq
    c_2\frac{|B_{2n}(w)|}{2^n}.$$
  A suitable numerical value for $c_2$ is $c_2=4.3$.
\end{lemma}

\begin{proof}
  From Proposition \ref{prop:K} it is enough to show that
  the sum $\sum_{C\in\C_n}\gamma^{n-\dim C}$ is upperbounded by a constant
  for any $\gamma<1$. Choosing a code $C$ in $\C_n$ is equivalent to
  choosing its generator polynomial, and from the list \eqref{eq:irreducible}
  of irreducible factors of $Z^n+1$, we see that if we order all possible
  generator polynomials by increasing degrees, we have $1$ and $1+Z$,
  then $2$ polynomials of degree at least $p-1$, then $4$ polynomials of
  degree at least $p(p-1)$, ... then $2^i$ polynomials of degree at least
  $p(p-1)^{i-1}$ and so on. Therefore, since
  $n-\dim C$ equals the degree of the generator polynomial, we obtain
  \begin{eqnarray*}
  \sum_{C\in\C_n}\gamma^{n-\dim C} 
   &\leq& 1+\gamma + 2\gamma^{p-1}+\sum_{i\geq 2}2^i\gamma^{p(p-1)^{i-1}} \\
    &\leq& 1+\gamma + 2\gamma^{p-1} 
     + \left(\frac{2}{p-1}\right)^2\sum_{i\geq 2}(p-1)^i\gamma^{(p-1)^{i}}\\
   &\leq& 1+\gamma + 2\gamma^{p-1}
      + \left(\frac{2}{p-1}\right)^2\sum_{j\geq 1}j\gamma^j\\
   &\leq& 1+\gamma + 2\gamma^{p-1} + \left(\frac{2}{p-1}\right)^2\frac{\gamma}{(1-\gamma)^2}.
  \end{eqnarray*}
With the values $\gamma=2^{1/5}$, $c_1=2^{6/5}$ and $p\geq 14^3$ given in 
Proposition~\ref{prop:K} we obtain that $c_2=4.3$ is suitable.
\end{proof}

From \eqref{eq:triplesum} and Lemma \ref{lem:c2} we obtain that
\begin{equation}\label{eq:c_2}
  \prob{X(w)>0}\leq c_2\frac 1n\frac{|B_{2n}(w)|}{2^n}+
  c_2\sum_{s=1}^{m-1}\frac{p^s}{n}\frac{|B_{2n/p^s}(w/p^s)|}{2^{n/p^s}}
\end{equation}
to deal with this last sum we invoke:

\begin{lemma}\label{lem:series}
  For any prime $p>14^3$ and 
  for any positive number $w$ such that $|B_{2n}(w)|\leq n2^n$, we have
  $$\sum_{s=1}^{m-1}\frac{p^s}{n}
    \frac{|B_{2n/p^s}(w/p^s)|}{2^{n/p^s}}\leq \frac 2p$$
\end{lemma}
\begin{proof}
  Choose $p$ times a vector of length $2n/p$ and weight not more than $w/p$:
  concatenate the resulting vectors and one obtains a vector of length $2n$ 
  and weight not more than $w$. 
  Therefore $|B_{2n/p}(w/p)|^p\leq |B_{2n}(w)|$ and
  we have
  $$\sum_{s=1}^{m-1}\frac{p^s}{n}
    \frac{|B_{2n/p^s}(w/p^s)|}{2^{n/p^s}}\leq
    \sum_{s=1}^{m-1}\frac{p^s}{n}\left(\frac{|B_{2n}(w)|}{2^n}\right)^{1/p^s}
    \leq \sum_{s=1}^{m-1}\frac{p^s}{n}n^{1/p^s}.$$
The result follows from routine computations.
\end{proof}

We see therefore from \eqref{eq:c_2} and Lemma \ref{lem:series}
that, if we choose $w$ such that $|B_{2n}(w)|\leq bn2^n$,
for $b<1$, then, provided the conditions of Proposition~\ref{prop:K}
are satisfied, we have
$\prob{X(w)>0}\leq bc_2 + 2c_2/p$. 
For $c_2=4.3$ and any $p>14^3$
this quantity is less than $1$ when $b\leq 0.23$. The largest $w$
for which $|B_{2n}(w)|\leq bn2^n$ is readily seen to satisfy
$K \leq \frac{w}{2n} \leq \frac 14$ which means that all conditions of
Proposition~\ref{prop:K} are satisfied, so that we have proved:

\begin{theorem}\label{th:main}
There exist positive constants $b\leq 0.23$ and $q$, 
such that for any prime $p\geq q$
such that $2$ is primitive modulo $p$ and $2^{p-1}\neq 1 \bmod p^2$,
and for any power $n=p^m$ of $p$, there exist double circulant codes
of parameters $[2n,n,d>w]$ for any $w$ such that $|B_{2n}(w)|\leq bn2^n$.
A suitable value of $q$ is $q=14^3$ and the first suitable prime $p$
is $p=2789$.
\end{theorem}

\subsection{Proof of Proposition \ref{prop:K}}\label{sec:proof}
Our remaining task is now to prove Proposition \ref{prop:K}. 
We start by noting that Proposition~\ref{prop:K} is stated with a
positive real number $w$, because the discussion starting 
from~ \eqref{eq:triplesum} involves balls of non-integer radius.
However, it clearly is enough to prove it only for integer values of $w$.

The crucial part of the proof will be
to bound from above the weight distribution of $C$, for $C\in\C_n$.
Let us note that, since the polynomial $P_n(Z)$ defined in
\eqref{eq:repetition}
is an irreducible factor of $Z^n+1$, the code $C$ belongs to $\C_n$ if and
only if $P_n(Z)$ divides the generator polynomial of the dual code $C^\perp$.
This means that any codeword of $C^\perp$ must be obtained by repeating
$p$ times a subvector of length $n/p$. Equivalently, a generating
matrix of $C^\perp$, i.e. a parity-check matrix of $C$ is of the form
  $$\H_C = [\A\;|\;\A\;|\cdots |\;\A]$$
meaning that it equals the concatenation of $p$ identical copies of an
$r\times n/p$ matrix~$\A$.

We shall need the following lemma.
\begin{lemma}\label{lem:repetition}
  Let $\H_{tr} = [\I_r\;|\; \I_r\; |\cdots |\;\I_r]$ be the $r\times tr$
  matrix obtained by concatenating $t$ copies of the $r\times r$ identity
  matrix. Let $\sigma_{tr}$ be the associated syndrome function:
  \begin{eqnarray*}
    \sigma_{tr} : \{0,1\}^{tr} & \rightarrow & \{0,1\}^{r}\\
          \x     & \mapsto     & \sigma_{tr}(\x) =\x\transp{\H_{tr}}.
  \end{eqnarray*}
  Let $w\leq tr$ be an integer. 
  Then, for any $\s\in\{0,1\}^{r}$, the number of vectors of length $tr$
  and of weight $w$ that map to $\s$ by $\sigma_{tr}$ is not more than:
    $$\sqrt{2rt}\left(\frac{1+|1-2\omega|^t}{2}\right)^r\binom{tr}{w}$$
  where $w=\omega tr$.
\end{lemma}

\begin{proof}
  Let $\X$ be a random vector of length $tr$ obtained by choosing independently
  each of its coordinates to equal $1$ with probability $\omega$.
  The probabilities that any given coordinate of $\sigma_{tr}(\X)$
  equals $0$ or $1$ are those of a sum of $t$ independent Bernoulli random
  variables of parameter $\omega$, namely:
   $$\frac{1+(1-2\omega)^t}{2}\hspace{1cm}\text{and}\hspace{1cm}
     \frac{1-(1-2\omega)^t}{2}.$$
  Since all the coordinates of $\sigma_{tr}(\X)$ are clearly independent,
  \begin{equation}\label{eq:max_s}
    \max_{\s\in\{0,1\}^r}\prob{\sigma_{tr}(\X)=\s}
    =\left(\frac{1+|1-2\omega|^t}{2}\right)^r.
  \end{equation}
  Now let $W=\wt{\X}$ be the weight of $\X$.  We have
  $$\prob{W=w}=\binom{tr}{w}\omega^{w}(1-\omega)^{tr-w}=
    \binom{tr}{\omega tr}2^{-trh(\omega)}$$
  where $h$ denotes the binary entropy function, 
  $h(x)=-x\log_2x-(1-x)\log_2(1-x)$. By a variant of
  Stirling's formula \cite{MW}[Ch. 10,\S 11,Lemma 7]
  \begin{equation}
    \label{eq:stirling}
    \binom nw \geq 2^{nh(\omega)}/\sqrt{8n\omega(1-\omega)},
  \end{equation}
  therefore:
    $$\prob{W=w}\geq \frac{1}{\sqrt{8tr\omega(1-\omega)}}
              \geq \frac{1}{\sqrt{2tr}}.$$
  For given $\s$, let $N_w$ denote the number of vectors of length $tr$ 
  and weight $w$ that have syndrome $\s$. Since 
  $\prob{\sigma_{tr}(\X)=\s\;|\; W=w} = N_w/\binom{tr}{w}$ we have 
  $$\prob{\sigma_{tr}(\X)=\s}\geq \prob{\sigma_{tr}(\X)=\s\;|\; W=w}\prob{W=w}
    \geq \frac{N_w}{\binom{tr}{w}}\frac{1}{\sqrt{2tr}}.$$
  Hence, by \eqref{eq:max_s},
  $$N_w\leq \sqrt{2tr}\left(\frac{1+|1-2\omega|^t}{2}\right)^r\binom{tr}{w}$$
  which is the claimed result.
\end{proof}

\begin{lemma}\label{lem:kappa}
  Let $0<\kappa <1/4$.
  There exist $q$, such that for any $p>q$, $n=p^m$, and for any
  code $C\in\C_n$, the following holds:
  \begin{itemize}
  \item either $C=\{0,1\}^n$ or $C$ equals the even-weight code,
  \item or, the weight distribution of $C$ satisfies, for any $i$, 
  $\kappa n\leq i\leq n/2$,
    $$A_i(C) \leq \frac{1}{2^{3r/5}}\binom ni$$
  where $r=n-\dim C$.
  \end{itemize}
  For $\kappa = 0.07$ a suitable value of $q$ is $q=14^3$.
\end{lemma}

\begin{proof}
  If $r=0$ or $r=1$, i.e. $C$ equals the whole space $\{0,1\}^n$ or the
  even-weight code, there is nothing to prove. Suppose therefore $r>1$.
  From the factorization \eqref{eq:irreducible} of $Z^n+1$ into irreducible 
  factors we see that we must have $r\geq p-1$. From the discussion preceding
  Lemma \ref{lem:repetition} we must have
  \begin{equation}
    \label{eq:r<=n/p}
    r\leq n-p^{m-1}(p-1)=n/p
  \end{equation}
  and a parity-check matrix of $C$ is made up of $p$ identical copies of
  some $r\times n/p$ matrix $\A$. Therefore, after permuting coordinates,
  there exists a parity-check matrix of $C$ of the form
    $$\H_C = [\B\;|\I_r\;|\; \I_r\; |\cdots |\;\I_r]$$
  where $\B$ is some $r\times (n-rt)$ matrix and is followed by $t$
  copies of the $r\times r$ identity matrix. The integer $t$ can be chosen 
  to take any value such that $1\leq t\leq p$: we shall impose the
  restriction
  \begin{equation}
    \label{eq:t<p^(1/3)}
    t\leq p^{1/3}.
  \end{equation}
  For any $\x\in\{0,1\}^n$, write $\x=(\x_1,\x_2)$ where $\x_1$ is the
  vector made up of the first $n-tr$ coordinates of $\x$ and $\x_2$ 
  consists of the remainding $tr$ coordinates
  Now the syndrome function $\sigma$ associated to $\H_C$
  takes the vector $\x\in\{0,1\}^n$ to 
  $\sigma(\x)=\x_1\transp{B}+\sigma_{tr}(\x_2)$ where $\sigma_{tr}$ is the
  function defined in Lemma \ref{lem:repetition}. The code~$C$ is the set
  of vectors $\x$ such that $\sigma(\x)=0$, therefore by partitioning
  the set of vectors of weight $i$ into all possible values of $\x_1$
  we have, from Lemma~\ref{lem:repetition}:
  \begin{equation}\label{eq:distrib}
  A_i(C)\leq\sqrt{2tr}
  \sum_{j=0}^{tr}
  \left(\frac{1+|1-2\frac{j}{tr}|^t}{2}\right)^r\binom{tr}{j}\binom{n-tr}{i-j}
  \end{equation}
for any $i$ such that
\begin{equation}
  \label{eq:i>tr}
  i\geq tr.
\end{equation}
Notice that:
 $$\binom{tr}{j}\binom{n-tr}{i-j} = 
  \frac{\binom ij \binom{n-i}{tr-j}}{\binom{n}{tr}}\binom ni$$
so that \eqref{eq:distrib} becomes
\begin{align}
   A_i(C)  &\leq\sqrt{2tr}
\sum_{j=0}^{tr}\left(\frac{1+|1-2\frac{j}{tr}|^t}{2}\right)^r
\frac{\binom ij \binom{n-i}{tr-j}}{\binom{n}{tr}}\binom ni\nonumber\\ 
           &\leq\sqrt{2tr}(tr+1)\binom ni\max_{0\leq j\leq tr}
          \left(\frac{1+|1-2\frac{j}{tr}|^t}{2}\right)^r
\frac{\binom ij \binom{n-i}{tr-j}}{\binom{n}{tr}}.\label{eq:distrib2}
\end{align}
Set $i= \iota n$ and $j=\alpha tr$, we have:
\begin{eqnarray*}
  \frac{\binom ij \binom{n-i}{tr-j}}{\binom{n}{tr}} 
& \leq &
\frac{i^j(n-i)^{tr-j}}{\binom{n}{tr}j!(tr-j)!}\\
& \leq & 
\frac{\iota^j(1-\iota)^{tr-j}n^{tr}}{\binom{n}{tr} j!(tr-j)!}\\
& \leq & 
\frac{\iota^j(1-\iota)^{tr-j}n^{tr}}{(n-tr)^{tr}\binom{tr}{j}^{-1}}
\hspace{1cm}\text{since $\tbinom{n}{tr}\geq (n-tr)^{tr}/(tr)!$}\\
& \leq & 
\frac{\iota^j(1-\iota)^{tr-j}\binom{tr}{j}}{(1-\frac{tr}{n})^{tr}}.
\end{eqnarray*}
We have seen \eqref{eq:r<=n/p} that $r\leq n/p$ and $t\leq p^{1/3}$ (condition
\eqref{eq:t<p^(1/3)}),
therefore $tr/n\leq 1/p ^{2/3}\leq 1/2$:
by using the inequality $1-x\geq 2^{-2x}$, valid whenever
$0\leq x\leq 1/2$, we therefore have
  $$\frac{\binom ij \binom{n-i}{tr-j}}{\binom{n}{tr}} \leq 
  2^{2t^2r^2/n}\iota^j(1-\iota)^{tr-j}\binom{tr}{j}$$
and by using $\binom{tr}{j}\leq 2^{trh(\alpha)}$ we finally get
$$\frac{\binom ij \binom{n-i}{tr-j}}{\binom{n}{tr}} \leq 
 2^{tr(\frac{2tr}{n}-D(\alpha ||\iota))}$$
where $D(x||y) = x\log_2\frac xy +(1-x)\log_2\frac{1-x}{1-y}$. Together with
\eqref{eq:distrib2} we get:
  $$A_i(C)\leq 2^{r(\beta +f(\iota))}\frac{1}{2^r}\binom ni$$
with
\begin{align}
  f(\iota) &= \max_{0\leq \alpha\leq 1}g(\alpha,\iota)\label{eq:f}\\
\text{where}\hspace{1cm}
  g(\alpha,\iota) &= \log_2(1+|1-2\alpha|^t)-tD(\alpha || \iota)\label{eq:g}
\end{align}
and $\beta = \frac 1r\log_2\sqrt{2tr} + \frac 1r\log_2(tr+1) + 2t^2r/n$.
Write $\log_2(tr+1)\leq 1+\log_2tr$ to get 
$\beta \leq (\frac 32 +\frac 32\log_2(tr))/r + 2t^2r/n$.
By using $t<p^{1/3}$ and $p-1\leq r\leq n/p$, we get
$$\frac 32\frac{\log_2tr}{r}<\frac 32\frac{\log_2(r+1)^{4/3}}{r}
  =2\frac{\log_2(r+1)}{r}\leq 2\frac{\log_2p}{p-1}$$
and
  $$\beta \leq \frac{3}{2(p-1)}+\frac{2\log_2p}{p-1}+\frac{2}{p^{1/3}}.$$
We see that $\beta$ can be made arbitrarily small by increasing the value
of $p$. A numerical computation gives us $\beta< 0.152$ 
for all $p>14^3$. 

Since we have supposed $i\leq n/2$, we have $\iota\leq 1/2$ so that
the definition \eqref{eq:f} and \eqref{eq:g} of $f$ can be replaced by the
equivalent
\begin{align*}
  f(\iota) &= \max_{0\leq \alpha\leq \iota}g(\alpha,\iota)\\
  g(\alpha,\iota) &= \log_2(1+(1-2\alpha)^t)-tD(\alpha || \iota)
\end{align*}
from which we easily see that $g$ and $f$ are decreasing
functions of $\iota$.
We see that $f(\kappa)$ can be made arbitrarily
small, for all $\kappa >0$, by choosing $t$ big enough. Numerically,
by choosing $t=14$, $\kappa =0.07$ and $p>14^3$, we see that \eqref{eq:i>tr}
is satisfied and we get, for all $0.07\leq\iota$,
  $f(\iota)\leq f(\kappa)\leq 0.24$. We obtain therefore that,
for all $\kappa n\leq i\leq n/2$,
  $$A_i(C)\leq 2^{-0.608r}\binom ni$$
which proves the lemma.
 \end{proof}

To prove Proposition \ref{prop:K}, we need a final technical lemma,
of a purely enumerative nature.

 \begin{lemma}\label{lem:enumeration}
   Let $0<\kappa < K <1/4$. There exist an integer $n_0$ and $\epsilon >0$
   such that, for any $n\geq n_0$, $w=2\omega n$ with $K\leq\omega <1/4$,
   $$2\sum_{\substack{i+j\leq w\\ i<\kappa n}}\binom ni \binom nj
     \leq \frac{1}{2^{\epsilon n}}|B_{2n}(w)|.$$
   For $\kappa = 0.07$, $K=0.1$, $n_0=14^3$, a suitable value of
   $\epsilon$ is $\epsilon =0.004$.
 \end{lemma}

 \begin{proof}
 Clearly we have:
   \begin{eqnarray*}
     2\sum_{\substack{i+j\leq w\\ i<\kappa n}}\binom ni \binom nj 
     &\leq& \kappa n^2\binom{n}{\kappa n}\binom{n}{w-\kappa n}\\
     &\leq& \kappa n^22^{n(h(\kappa)+h(2\omega -\kappa))} = 
      \frac{\kappa n^22^{2nh(\omega)}}
      {2^{n(2h(\omega)-h(\kappa)-h(2\omega -\kappa))}}\\
     &\leq& \kappa n^2
      \frac{2^{2nh(\omega)}}{2^{n(2h(K)-h(\kappa)-h(2K -\kappa))}}
   \end{eqnarray*}
since $2h(\omega)-h(\kappa)-h(2\omega -\kappa)$ is an increasing 
function of $\omega$.
By \eqref{eq:stirling} we have $2^{2nh(\omega)}\leq \sqrt{16n}|B_{2n}(w)|$,
so that we obtain, since $\kappa\leq 1/4$,
  $$2\sum_{\substack{i+j\leq w\\ i<\kappa n}}\binom ni \binom nj \leq
    n^{5/2}\frac{|B_{2n}(w)|}{2^{n(2h(K)-h(\kappa)-h(2K -\kappa))}}
    \leq \frac{|B_{2n}(w)|}{2^{\epsilon n}}$$
  for any $n\geq n_0$ with 
  $\epsilon \leq 2h(K)-h(\kappa)-h(2K -\kappa)-\frac 52\frac{\log_2n_0}{n_0}$.
 \end{proof}

 \begin{proofof}{Proposition}{\ref{prop:K}}
  If $C=\{0,1\}^n$ or if $C$ is the even-weight subcode, then
  $A_i(C)\leq \binom ni$, and $\sum_{i+j\leq w}A_i(C)A_j(C)\leq 
  \sum_{i+j\leq w}\binom ni\binom nj = |B_{2n}(w)|$. The result clearly
  holds for any $c_1\geq 2/\gamma$.

  Let $C\in\C_n$ with $r=n-\dim C>1$. Let us write:
  $$\frac{1}{|C|}\sum_{i+j\leq w}A_i(C)A_j(C) = S_1 + S_2$$
  with
  $$S_1=\frac{1}{|C|}\sum_{\substack{i+j\leq w\\ \kappa n\leq i,j}}A_i(C)A_j(C)
    \hspace{1cm}\text{and}\hspace{1cm}
    S_2=\frac{2}{|C|}\sum_{\substack{i+j\leq w\\ i<\kappa n}}A_i(C)A_j(C).$$
  By Lemma \ref{lem:kappa} we have
  $$S_1\leq \frac{1}{|C|}\sum_{i+j\leq w}\binom ni\binom nj\frac{1}{2^{6r/5}}
    \leq \frac{|B_{2n}(w)|}{2^n}\frac{1}{2^{r/5}}.$$
To upperbound $S_2$ we simply write $A_i(C)\leq \binom ni$. 
By Lemma~\ref{lem:enumeration}, we have
  $$S_2\leq \frac{|B_{2n}(w)|}{2^n}\frac{2^r}{2^{\epsilon n}}
       =\frac{|B_{2n}(w)|}{2^n}\frac{2^r}{(2^{\epsilon p})^{n/p}}
       \leq \frac{|B_{2n}(w)|}{2^n}\frac{2^r}{(2^{\epsilon p})^{r}}$$
since we have seen \eqref{eq:r<=n/p} that $r\leq n/p$. 
By choosing $p\geq \frac{6}{5\epsilon}$ we obtain
  $$S_2\leq \frac{|B_{2n}(w)|}{2^n}\frac{1}{2^{r/5}}.$$
This proves the result with $\gamma = 1/2^{1/5}$ and $c_1=2^{6/5}$.
 \end{proofof}

\section{Comments}\label{sec:more}

The probabilistic method we used easily shows that almost all
double circulant codes of the asymptotic family presented here satisfy
an improved bound of the form \eqref{eq:GV+}. Actually we suspect that
this is also the case for most choices of $n$~: this is suggested by
computer experiments with randomly chosen double circulant codes of
small blocklengths.

We have tried to strike a balance between giving readable proofs and
deriving a non-astronomical lower bound on the prime $p$ in 
Theorem~\ref{th:main}. In principle, the numerical values could be
refined. In particular, the constant $b$ of Theorem \ref{th:main}
could be made to approach  $1/2$ (as in Theorem~\ref{th:simple})
but at the cost of a larger~$p$.
If we convert the formulation of Theorem \ref{th:main} in the form
\eqref{eq:GV+} (which just involves switching from $|B_{2n}(d)|$ in
Theorem~\ref{th:GVplus} to $|B_{2n}(d-1)|$ in \eqref{eq:GV+}) 
we obtain a constant $c$ which is of the
same order of magnitude, but somewhat worse,
than the improved constant $c\approx 0.102$ of \cite{vw05} 
for Jiang and Vardy's method.

In this paper we only consider the binary case with codes
of rate $1/2$ but the method can 
be straightforwardly generalized
to the case of different alphabets and to quasi-cyclic codes 
of any rational rate (though at the cost of a worsening of the
constant $b$) by considering for parity check matrices
vertical and horizontal concatenations of random circulant matrices. 

Finally, a  natural question is to wonder whether the ideas developed in this
paper can be extended to Euclidean lattices in a way similar to the
generalization of Jiang and Vardy's method to sphere-packings of Euclidean
spaces \cite{klv04}.
A positive answer to this question is given in the paper \cite{gz}.

\end{document}